\documentclass[11pt]{article}
\usepackage{a4wide}         
\usepackage{amsmath,amsfonts,amssymb,amstext,amsthm}
\usepackage{filecontents} 
\usepackage{cite}
 
\usepackage{tikz}
\usetikzlibrary{decorations.markings}
\usetikzlibrary{arrows} 
\usepackage{color}
 
\def\calM{\mathcal{M}}
\def\iso{i} 
 \def\Tabs{T_{\text{A}}}
 \newcommand{\Tabsj}[1]{T_{\text{A},{#1}}}  
 \newcommand{\ceil}[1]{\left\lceil {#1}\right\rceil}
\newcommand{\floor}[1]{\left\lfloor {#1}\right\rfloor}
 
\newcommand{\StartClique}{\ensuremath{\mathcal{S_C}}}
\newcommand{\SometimePlentyClique}{\ensuremath{\mathcal{H_C}}}
\newcommand{\PlentyClique}[1]{\ensuremath{\mathcal H_{C,#1}}}
\newcommand{\Fixation}{\ensuremath{\mathcal{F}}}
\newcommand{\Extinction}{\ensuremath{\mathcal{E}}}
\newcommand{\EmptyClique}{\ensuremath{\mathcal{E_C}}} 
\newcommand{\Quick}{\ensuremath{\mathcal{Q}}}
\newcommand{\Slow}{\ensuremath{\mathcal{S}}}
\let\epsilon=\varepsilon
\renewcommand{\deg}{d}
\newcommand{\degplus}[1]{m^+_{#1}} 
\newcommand{\degminus}[1]{m^-_{#1}}

\newcommand{\Yt}{\widetilde{Y}}  
\newcommand{\Zt}{\widetilde{Z}}

\newcommand{\Exp}{\mathrm{Exp}}

\newcommand{\Tdisc}{T_{\mathrm{d}}} 
\newcommand{\Tcont}{T_{\mathrm{c}}}
\newcommand{\tFixation}{\widetilde{\Fixation}}

\newtheorem{theorem}{Theorem} 
\newtheorem{lemma}[theorem]{Lemma}
\newtheorem{corollary}[theorem]{Corollary}

\theoremstyle{definition}

\newtheorem{observation}[theorem]{Observation}   
\newtheorem{proposition}[theorem]{Proposition}

\makeatletter
\newtheorem*{rep@theorem}{\rep@title}
\newcommand{\newreptheorem}[2]{%
\newenvironment{rep#1}[1]{%
 \def\rep@title{#2 \ref{##1}}%
 \begin{rep@theorem}}%
 {\end{rep@theorem}}}
 \newtheorem*{rep@corollary}{\rep@title}
\newcommand{\newrepcorollary}[2]{%
\newenvironment{rep#1}[1]{%
 \def\rep@title{#2 \ref{##1}}%
 \begin{rep@corollary}}%
 {\end{rep@corollary}}}
\makeatother

\newreptheorem{theorem}{Theorem}
\newreptheorem{lemma}{Lemma}
\newreptheorem{corollary}{Corollary}

\title{Absorption Time of the Moran Process\thanks{
The research leading to these results has received funding from the European Research Council under the European Union's Seventh Framework Programme (FP7/2007--2013) ERC grant agreement no.\ 334828. The paper 
reflects only the authors' views and not   the views of the ERC or the European Commission. The European Union is not liable for any use that may be made of the information contained therein.}}
\author{Josep D\'{\i}az\thanks{Departament de Llenguatges i Sistemes
    Inform{\`a}tics, Universitat Polit{\`e}cnica de Catalunya, Spain.}
  \and Leslie Ann Goldberg\thanks{Department of Computer Science,
    University of Oxford, UK.}  \and David Richerby\footnotemark[3]
  \and Maria Serna\footnotemark[2]}

\begin{document}

\maketitle{}

\begin{abstract}

The Moran process models the spread of  mutations
in populations on graphs.  We investigate the
absorption time of the process, which is the time taken for a mutation
introduced at a randomly chosen vertex to either spread to the whole
population, or to become extinct.  It is known that the expected
absorption time for an advantageous mutation is $O(n^4)$ on an
$n$-vertex undirected graph, which allows the behaviour of the process
on undirected graphs to be analysed using the Markov chain Monte Carlo method.
We show that this does not extend to directed graphs by  exhibiting an infinite
family of directed graphs  for which the expected absorption time is
exponential in the number of vertices. However, for regular directed
graphs, we show that the expected absorption time is $\Omega(n\log n)$
and $O(n^2)$.  We exhibit families of graphs matching these bounds
and give improved bounds for other families of graphs, based on isoperimetric number.
Our results are obtained via stochastic dominations
which we demonstrate by establishing a coupling in a related
continuous-time model. The coupling also implies several natural domination results
regarding the fixation probability of the original (discrete-time) process,
resolving a conjecture of Shakarian, Roos and Johnson.

\end{abstract}

\section{Introduction}
\label{sec:intro}

The Moran process~\cite{Moran58}, as adapted by Lieberman, Hauert and
Nowak~\cite{LHN}, is a stochastic model for the spread of genetic
mutations through populations of organisms.  Similar processes have
been used to model the spread of epidemic diseases, the behaviour of
voters, the spread of ideas in social networks, strategic interaction
in evolutionary game theory, the emergence of monopolies, 
and cascading
failures in power grids and transport networks
\cite{ARLV2001:Influence, Ber2001:Monopolies, Gin2000:GT,
  KKT2003:Influence, Lig1999:IntSys}.

In the Moran process, individuals are modelled as the vertices of a
graph and, at each step of the discrete-time process, an individual is
selected at random to reproduce.  This vertex chooses one of its neighbours
uniformly at random and replaces that neighbour with its offspring, a
copy of itself.  The probability that any given individual is chosen
to reproduce is proportional to its \emph{fitness}: individuals with
the mutation have fitness~$r>0$ and non-mutants have fitness~$1$.  The
initial state has a single mutant placed uniformly at random in the
graph, with every other vertex a non-mutant.  
On
any finite, strongly connected graph, the
process will terminate with probability~$1$, either in the state where
every vertex is a mutant (known as \emph{fixation}) or 
in the state where no vertex
is a mutant (known as \emph{extinction}).

The principal quantities of interest are the \emph{fixation
probability} (the probability of reaching fixation) and the expected
\emph{absorption time} (the expected number of steps before fixation
or extinction is reached).  In general, these depend on both the
graph topology and the mutant fitness.  In principle, they can be
computed by standard Markov chain techniques but doing so for an
$n$-vertex graph involves solving a set of $2^n$~linear equations,
which is computationally infeasible.
Fixation
probabilities have also been calculated by producing and approximately
solving a set of differential equations that model the
process~\cite{HV2011:fixation}.  These methods seem to work well in
practice but there is no known bound on the error introduced by the
conversion to differential equations and the approximations in their
solution.

When the underlying graph is undirected and the mutant fitness~$r$ is at least~$1$,
it is known how to approximate the fixation probability:
The paper~\cite{OurSODA} gave a  fully polynomial
randomised approximation scheme (FPRAS) for computing the fixation
probability of the Moran process on undirected graphs.
The approximation scheme uses
the Markov chain Monte Carlo method.
The fact that it provides a suitable approximation in polynomial time
follows from the fact that  the
expected absorption time on an $n$-vertex graph is at most
$\tfrac{r}{r-1} n^4$ for $r>1$.

\subsection{Our contributions}
 
The main contribution of this paper is to determine the extent to
which the polynomial  
bound on expected absorption time 
carries through to directed graphs. 
Throughout the paper, we assume that the mutant fitness~$r$  exceeds~$1$.

\subsubsection{Regular digraphs}

We start by considering the absorption time on a 
strongly connected $\Delta$-regular digraph
(where every vertex has in-degree~$\Delta$ and out-degree~$\Delta$).
Regularity makes some calculations straightforward 
because the detailed topology of the graph is not relevant.
We describe these first, and then discuss the more
difficult questions (where topology does play a role) and
state our results.

The following facts hold for $\Delta$-regular graphs, independent of
the topology.
\begin{itemize}
\item
It is well known~\cite{LHN} that, on any regular graph on $n$~vertices,
a single randomly placed mutant with fitness~$r$ reaches
fixation with probability
\begin{equation}
\label{eq:isothermal}
    \frac{1-r^{-1}}{1-r^{-n}}
\end{equation}
To see this, note that if there are $k$~mutants, the total fitness of the
population is $W_k = n+k(r-1)$.  The probability that the next
reproduction happens along the directed edge $(u,v)$ is
$\tfrac{r}{W_k}\tfrac1\Delta$ if $u$~is a mutant and
$\tfrac1{W_k}\tfrac1\Delta$, if it is not.  Since the graph is
$\Delta$-regular, there are   exactly as many directed edges
from mutants to non-mutants as there are from non-mutants to mutants.
It follows that the probability that the number of mutants increases
at the next step is exactly $r$~times the probability that it
decreases, regardless of which vertices are currently mutants.  Thus,
the number of mutants in the population, observed every time it
changes, forms a random walk on $\{0, \dots, n\}$ with initial
state~$1$, upward drift~$r$ and absorbing barriers at $0$ and~$n$.  It
is a standard result (e.g., \cite[Example~3.9.6]{GS2001:Probability})
that such a random walk reaches the barrier at~$n$ with the
probability given in~\eqref{eq:isothermal}.
\item
It is also well known (e.g., \cite[Example~3.9.6]{GS2001:Probability}) 
that the expected number of steps of this walk before
absorption (which may be at either $0$ or~$n$) is a function of $r$
and~$n$ that tends to $n(1+\tfrac1r)$ in the limit as $n\to\infty$,
independent of the graph structure beyond regularity.
  However,
the number of steps taken by the random walk (often referred to as the
``active steps'' of the Moran process) is not the same as the original
process's absorption time, because the absorption time includes many steps at which
the number of mutants does not change, either because a mutant
reproduces to a mutant or 
because a non-mutant reproduces to a non-mutant.
\end{itemize}

In Section~\ref{sec:regular}, we show that the expected absorption
time of the Moran process is polynomial for regular digraphs.  In
contrast to the number of active steps, the absorption time does
depend on the detailed structure of the graph.
We prove the following  upper and lower bounds, where $H_n$ denotes the $n$'th
harmonic number, which is $\Theta(\log n)$.

\newcommand{\THMregbounds}{
The expected absorption time of the Moran process on a strongly connected $\Delta$-regular
$n$-vertex digraph $G$ is at least $\left(\frac{r-1}{r^2}\right)n H_{n-1}$ and at most
$n^2 \Delta$.}
\begin{theorem}\label{thm:regbounds}\THMregbounds{}
\end{theorem}

In Section~\ref{sec:family}, we prove the following
theorem, which shows that 
the upper bound in Theorem~\ref{thm:regbounds} is
tight up to a constant factor (which  depends on~$\Delta$ and~$r$ but not
on~$n$).

\newcommand{\THMtight}{
Suppose that $r>1$ and $\Delta>2$.
There is an infinite family~$\mathcal{G}$ of $\Delta$-regular graphs
such that,  when the Moran process is run on an $n$-vertex graph $G\in \mathcal{G}$,
 the  
    expected absorption time  
    exceeds $\tfrac{1}{8r}(1-\tfrac{1}{r})\frac{n^2}{(\Delta-1)^2} $.
}
\begin{theorem} \label{thm:tight}\THMtight{}
\end{theorem}

The digraphs in
the family~$\mathcal{G}$ are symmetric, so can be viewed as undirected graphs.
The upper bound on the expected absorption time in
Theorem~\ref{thm:regbounds} can be improved for certain classes of
regular undirected graphs using 
the notion of the isoperimetric number $\iso(G)$ of a graph~$G$,
which is defined 
in Section~\ref{sec:iso}.
There, we prove the following theorem.

\newcommand{\THMisoperimetric}{
The expected absorption time of the Moran process on a connected $\Delta$-regular
$n$-vertex undirected graph~$G$ is
at most 
$ 2\Delta n  H_n /\iso(G)$. 
}
\begin{theorem}
\label{thm:isoperimetric}
\THMisoperimetric{}
\end{theorem}

 Theorem~\ref{thm:isoperimetric} 
pinpoints the expected 
absorption time for $G=K_n$, up to a constant factor, since $\iso(K_n)=\lceil n/2\rceil$
\cite{Mohar} and  Theorem~\ref{thm:regbounds} gives an $\Omega(n \log n)$ lower bound.
 Theorem~\ref{thm:isoperimetric}
is worse than the upper bound of Theorem~\ref{thm:regbounds} by
a factor of $O(\log n)$ for the cycle~$C_n$
since $\iso(C_n)=2/\lfloor n/2\rfloor$ \cite{Mohar}.
However, we often get an improvement by using the isoperimetric number.
For example,  
the $\sqrt{n}$-by-$\sqrt{n}$ grid
has   $\iso(G) = \Theta(1/\sqrt{n})$ (see~\cite{arrays}),
giving an $O(n^{3/2} \log n)$ absorption time; the hypercube
has  $\iso(G)=1$ (see, e.g., \cite{Mohar}), giving
an $O(n \log^2 n)$ absorption time.  Bollob{\'a}s~\cite{bollobas} showed that,
 for every $\Delta\geq 3$ 
there is a positive number~$\eta<1$ such that, for 
almost all $\Delta$-regular
$n$-vertex undirected
graphs~$G$ (as $n$ tends to infinity),
$\iso(G)\geq (1-\eta) \Delta/2$, which gives an $O(n \log n)$ absorption time
since these graphs are connected.

\subsubsection{Slow absorption}

Theorem~\ref{thm:regbounds} shows  that regular digraphs, like undirected graphs, reach
absorption in expected polynomial time.  
In Section~\ref{sec:directed} we
show that  the same does not hold for general digraphs.  In
particular, we construct an infinite
family $\{G_{r,N}\}$ of strongly connected digraphs
indexed by a positive integer~$N$.
We then prove the following theorem.

\begin{theorem}
\label{thm:directed}
Fix $r>1$ and let $\epsilon_r = \min(r-1,\,1)$.
For any positive integer~$N$ that is sufficiently large with respect to~$r$,
the expected absorption time of the Moran process on $G_{r,N}$ is
at least 
$$  \frac{1}{16}\floor{\left(\frac{\epsilon_r}{32}\right)(2^N-1)}
 \frac{\epsilon_r}{4\ceil{r}+3}.$$
\end{theorem}

Theorem~\ref{thm:directed} shows that there is an infinite
family of strongly connected digraphs such that the
absorption time on $n$-vertex graphs in this family is exponentially
large, as a function of~$n$.  It follows that the techniques  from~\cite{OurSODA} do not give a polynomial-time algorithm for approximating
the fixation probability on digraphs.

The underlying structure of the graph $G_{r,N}$ is a large undirected clique on
$N$~vertices and a long directed path.  Each vertex of the clique
sends an edge to the first vertex of the path, and each 
vertex of the clique receives an
edge from the path's last vertex.  We refer to the first~$N$ vertices
of path as~$P$ and the remainder as~$Q$.  Each vertex of~$P$ has
out-degree~$1$ but receives $4\ceil{r}$~edges from~$Q$.  (See
Figure~\ref{fig:GrN}.)

Suppose that $N$ is sufficiently large with respect to~$r$
and consider the Moran process on~$G_{r,N}$. 
Given the relative sizes of the clique and the path,
there is a reasonable
probability (about $\tfrac{1}{4r+2}$) that the initial mutant is in
the clique.  The edges to and from the path have a negligible effect
so it is reasonably likely (probability at least $1-\tfrac1r$) that we
will then reach the state where half the clique vertices are mutants.
To reach absorption from this state, one of two things must happen.

For the process to reach extinction, the mutants already in the clique
must die out.  Because the interaction between the clique and path is
small, the number of mutants in the clique is very close to a random
walk on $\{0, \dots, N\}$ with upward drift~$r$, and the expected time
before such a walk reaches zero from $N/2$ is exponential in~$N$.

On the other hand, suppose the process reaches fixation.  In
particular, the vertices of~$P$, the first part of the path, must
become mutants.  Note that no vertex of $Q$~can become a mutant before
the last vertex of~$P$ has done so.  While all the vertices in~$Q$ are
non-mutants, the edges from that part of the path to~$P$ ensure that
each mutant in $P$ is more likely to be replaced by a non-mutant
from~$Q$ than it is to create a new mutant in~$P$.  As a result, the
number of mutants in~$P$ is bounded above by a random walk on $\{0,
\dots, N\}$ with a strictly greater probability of decreasing than
increasing.  Again, this walk is expected to take exponentially many
steps before reaching~$N$.

\subsubsection{Stochastic domination}
\label{sec:introdom}

Our main technical tool is stochastic domination.  Intuitively, one
expects that the Moran process has a higher probability of reaching
fixation when the set of mutants is~$S$ than when it is some subset
of~$S$, and that it is likely to do so in fewer steps.  It also seems obvious
that modifying the process by continuing to allow all transitions that create new
mutants but forbidding some transitions that remove mutants should
make fixation faster and more probable.  Such intuitions have
been used in proofs in the literature; it turns out that they are essentially
correct, but for rather subtle reasons.

The Moran process can be described as a Markov chain $(Y_t)_{t\geq 1}$
where $Y_t$ is the set $S\subseteq V(G)$ of mutants at the $t$'th step.
The normal method to make  the  above intuitions formal would be to 
demonstrate a stochastic domination by coupling
the Moran process $(Y_t)_{t\geq 1}$  
with another copy $(Y'_t)_{t\geq 1}$
of the process where $Y_1 \subseteq Y'_1$.
The coupling would be designed so that
$Y_1\subseteq Y'_1$
would ensure that  
$Y_t\subseteq Y'_t$ for all $t>1$.  However, a simple example shows
that such a  coupling does not always exist for the Moran process. 
Let $G$ be the undirected path with two
edges: $V(G)=\{1,2,3\}$ and $E(G)=\{(1,2), (2,1), (2,3), (3,2)\}$.
Let $(Y_t)_{t\geq 1}$ and $(Y'_t)_{t\geq 1}$ be Moran processes on~$G$
with $Y_1=\{2\}$ and $Y'_1 = \{2,3\}$.  With probability
$\tfrac{r}{2(r+2)}$, we have $Y_2 = \{1,2\}$.  The only possible value
for~$Y'_2$ that contains~$Y_2$ is $\{1,2,3\}$ but this occurs with
probability only $\tfrac{r}{2(2r+1)}$.  Therefore, any coupling
between the two processes fails because
\begin{equation*}
    \Pr(Y_2\not\subseteq Y'_2) \geq \frac{r(r-1)}{2(r+2)(2r+1)}\,,
\end{equation*}
which is strictly positive for any $r>1$.  The problem is that, when
vertex~$3$ becomes a mutant, it becomes more likely to be the next
vertex to reproduce and, correspondingly, every other vertex becomes
less likely.  This can be seen as the new mutant ``slowing down'' all
the other vertices in the graph.

To get around this problem, we consider a continuous-time version of
the process, $\Yt[t]$ ($t\geq 0$).  Given the set of mutants~$\Yt[t]$
at time~$t$, each vertex waits an amount of time before reproducing.
For each vertex, this period of time is chosen according to the
exponential distribution with parameter equal to the vertex's fitness,
independently of the other vertices.  (Thus, the parameter is $r$ if
the vertex is a mutant and $1$, otherwise.)  If the first vertex to
reproduce is~$v$ at time~$t+\tau$ then, as in the standard,
discrete-time version of the process, one of its out-neighbours~$w$ is
chosen uniformly at random, the individual at~$w$ is replaced by a
copy of the one at~$v$ and the time at which $w$ will next reproduce is
exponentially distributed with parameter given by its new fitness.
The discrete-time process is recovered by taking the sequence of
configurations each time a vertex reproduces.

In continuous time, each member of the population reproduces at a rate
given by its fitness, independently of the rest of the population
whereas, in discrete time, the population has to co-ordinate to decide
who will reproduce next.  It is still true in continuous time that
vertex~$w$ becoming a mutant makes it less likely that each vertex
$v\neq w$ will be the next to reproduce.  However, the vertices are
not slowed down as they are in discrete time: they continue to
reproduce at rates determined by their fitnesses.  This distinction
allows us to establish the following coupling lemma,
which formalises the intuitions discussed above. 

\newcommand{\LEMMAcouple}{
    Let $G=(V,E)$ be any digraph, let $Y\subseteq Y'\subseteq V(G)$ and
    $1\leq r\leq r'\!$.  Let $\Yt[t]$ and $\Yt'[t]$ ($t\geq 0$) be
    continuous-time Moran processes on~$G$ with mutant fitness $r$
    and~$r'\!$, respectively, and with $\Yt[0]=Y$ and $\Yt'[0]=Y'\!$.
    There is a coupling between the two processes such that
    $\Yt[t]\subseteq \Yt'[t]$ for all $t\geq 0$.
}
\begin{lemma}[Coupling lemma]
\label{lemma:couple}
\LEMMAcouple{}
\end{lemma}

In the paper, we use the coupling lemma to establish 
stochastic dominations between 
\emph{discrete} Moran
processes. It also has consequences
concerning fixation probabilities. Recall that
``fixation'' is the state of a (discrete) Moran process in which every vertex is a mutant. 
The fixation probability
$f_{G,r}$ is the probability that this state is reached
when  the Moran process is run on a digraph~$G=(V,E)$, 
starting from a state in which 
a  single
initial mutant is placed uniformly at random.  
 For any set $S\subseteq
V$, let $f_{G,r}(S)$ be the probability of reaching fixation when
the set of vertices initially occupied by mutants is~$S$.   Thus,
$f_{G,r} = \tfrac{1}{|V|}\sum_{v\in V} f_{G,r}(\{v\})$.
Using the coupling lemma, we can prove the following theorem.

\newcommand{\THMfix}{
    For any digraph $G$, if $0<r\leq r'$ and $S\subseteq
    S'\subseteq V(G)$, then $f_{G,r}(S) \leq f_{G,r'}(S')$.
}
\begin{theorem}
\label{thm:fix}
\THMfix{}
\end{theorem}

This theorem has two immediate corollaries.  The first was conjectured
by Shakarian, Roos and 
Johnson~\cite[Conjecture 2.1]{SRJ2012:review}.  
It was known
from~\cite{SR2011:fastfix} that $f_{G,r}\geq f_{G,1}$ for any $r>1$.

\begin{corollary}[Monotonicity]\label{cor:monotonicity}
    If $0<r\leq r'$ then, for any digraph $G$, $f_{G,r}\leq f_{G,r'}$.
\end{corollary}

Corollary~\ref{cor:monotonicity} follows
immediately from Theorem~\ref{thm:fix} since $f_{G,r}(\{v\})\leq f_{G,r'}(\{v\})$ for all $v\in V(G)$.

The second corollary can be stated informally as, ``Adding more
mutants can't decrease the fixation probability,'' and has been
assumed in the literature, without proof.  However, although it
appears obvious at first, it is somewhat subtle: the example
at the beginning of this section shows that adding more mutants can actually
decrease the probability of a particular vertex becoming a mutant at
the next step of the process.

\begin{corollary}[Subset domination]
\label{cor:subsetdom}
    For any digraph $G$ and any $r>0$, if $S\subseteq S'\subseteq
    V(G)$, then $f_{G,r}(S)\leq f_{G,r}(S')$.
\end{corollary}
 
Note that, although we have introduced the continuous-time version of the process
for technical reasons, to draw conclusions about the original,
discrete-time Moran process, the continuous-time version
may actually be a more realistic model than the discrete-time version.

 \subsection{Previous work}
 
 There is previous work on calculating the fixation probability
 of the Moran process. 
Fixation probabilities are known for 
regular graphs~\cite{LHN} and for stars (complete
bipartite graphs $K_{1,k}$)~\cite{BR2008:fixprob}.  Lieberman et
al.~\cite{LHN} have defined classes of directed graphs with a
parameter~$k$,  
for which  they claim that the fixation probability tends to
$1-r^{-k}$ for large graphs. While these graphs
do seem to have very large fixation probability,   we have shown this specific claim
to be incorrect for $k=5$~\cite{DGMRSS2013:superstars}.  
Very recently, it has been claimed~\cite{fixrevised}
that for large~$k$, the fixation probability is close to $1-\frac{1}{(k-2)r^4}$.
Other work
has investigated the possibility of so-called ``suppressors'', graphs
having fixation probability less than that given
by~\eqref{eq:isothermal} for at least some range of values for~$r$
\cite{MNRS2013:suppressors, MS2013:strongbounds}.

There is a more complicated version of the Moran process 
in which the fitness
of a vertex is determined by its expected payoff when playing some
two-player game against a randomly chosen
neighbour~\cite{TFSN2004:game, SRJ2012:review}.  In this version of
the process, mutants play the game with one strategy and 
non-mutants play the game with another.  
The ordinary Moran process
corresponds to the special case of
this game in which the mutant and non-mutant strategies
give payoffs $r$ and~$1$, respectively, regardless of the strategy
used by the opponent.

Most previous work on absorption times has been in the game-based version of
the process, where the added complexity of the model
limits analysis to very simple graphs, such as complete graphs, stars
and cycles~\cite{AS2006:fixtime,
  BHR2010:speed}.  

\section{Preliminaries}

When $k$ is a positive integer, $[k]$ denotes $\{1,\ldots,k\}$.
We consider the evolution of the Moran process~\cite{LHN}
on a 
strongly connected directed
graph (digraph).
Consider 
such
a  digraph  $G=(V,E)$ with $n=|V|$.
When the process is run on~$G$, each state is a set of vertices $S\subseteq V$.
The vertices in~$S$ are said to be ``mutants''.
If $|S|=k$ then the total fitness of the state  is given by  
$W_k=n+(r-1)k$  --- each of the $k$ mutants contributes fitness~$r$
to the total fitness
and each other vertex contributes fitness~$1$.  Except where stated
otherwise, we assume that $r>1$.
The starting state is chosen uniformly at random from the one-element subsets of~$V(G)$.
From a state~$S$ with $|S|=k$,
the process evolves as follows.
First, a vertex~$u$ is chosen to reproduce.
The probability that vertex~$u$ is chosen is 
$r/W_k$ if $u$ is a mutant and $1/W_k$ otherwise.
Given that $u$ will reproduce,
a directed edge $(u,v)$ is chosen uniformly at random from 
$\{(u,v') \mid (u,v')\in E\}$.  The state of vertex~$u$ in~$S$ is copied
to vertex~$v$ to give the new state
\begin{equation*}
    S|_{u\to v} = \begin{cases}
                \ S\cup\{v\}      & \text{if $u$ is a mutant,} \\
                \ S\setminus\{v\} & \text{if $u$ is a non-mutant.}
                \end{cases}
\end{equation*}

Let $\deg^+(u)$ denote the out-degree of vertex~$u$
and let $\deg^-(u)$ denote its in-degree.
A digraph~$G=(V,E)$ is \emph{$\Delta$-regular}
if, for every vertex $u\in V$, $\deg^+(u)=\deg^-(u)=\Delta$.
$G$ is \emph{regular} if it is $\Delta$-regular for some
natural number~$\Delta$.
If the Moran process is run on a strongly connected
digraph~$G$, there are exactly two absorbing states --- $\emptyset$ and~$V(G)$.
Once one of these states is reached, the process will stay in it forever. The
\emph{absorption time} is the number of reproduction steps until
such a state is reached.

A digraph~$G=(V,E)$ is weakly connected if the underlying undirected
graph is connected. Given a subset $S\subseteq V$,
let $\degplus S$ be the number of edges 
from vertices in~$S$ to vertices in~$V\setminus S$.
Let $\degminus S$ be the number of edges from vertices in $V\setminus S$ to~$S$.
If $G$ is regular then, for every~$S\subseteq V$,
\begin{align*}
    \degplus S &= |\{(u,v)\in E\mid u\in S\}|
                      - |\{(u,v)\in E\mid u,v\in S\}| \\
               &= |S|\Delta - |\{(u,v)\in E\mid u,v\in S\}| \\
               &= |\{(u,v)\in E\mid v\in S\}|
                      - |\{(u,v)\in E\mid u,v\in S\}| \\
               &= \degminus S\,.
\end{align*}
Thus, every regular digraph that is weakly connected is strongly
connected.

We sometimes consider the Moran process on an undirected graph~$G=(V,E)$. 
We view the undirected graph as a digraph in which 
the set~$E$ of edges is symmetric
(so $(u,v)\in E$ if and only if $(v,u)\in E$).  
If $G$ is undirected then, for every vertex $u$, 
$\deg^+(u)=\deg^-(u)$ and in this case we just write $\deg(u)$ to denote this quantity.

\section{Domination}
\label{sec:domination}

A useful proof technique is to show that the behaviour of the Moran
process is stochastically dominated by that of a related process that is
easier to analyse.  
Similarly, it is useful to compare the behaviour of the Moran process, evolving
on a digraph~$G$, with that of another Moran process on the same digraph,
where the second process starts with more mutants.
Recall that the Moran process can be described as a Markov chain $(Y_t)_{t\geq 1}$
where $Y_t$ is the set $S\subseteq V(G)$ of mutants at the $t$'th step.
It would be natural to attempt to
establish a coupling between Moran processes $(Y_t)_{t\geq 1}$ and
$(Y'_t)_{t\geq 1}$ such that, if $Y_1\subseteq Y'_1$, then
$Y_t\subseteq Y'_t$ for all $t\geq 1$
but, as we showed in Section~\ref{sec:introdom}, this
cannot be done.
 
To obtain useful dominations, we will consider
a continuous-time version of the Moran process.
The   domination that we construct for the continuous-time process
will allow us to draw conclusions about
the original (discrete-time) Moran process.
In a digraph $G=(V,E)$ where
the set of mutants~$Y$ have fitness~$r$, let $r_{v,Y}=r$ if $v\in Y$
and $r_{v,Y}=1$, otherwise.
We define the continuous-time version of the 
Moran process
on a digraph $G=(V,E)$ as follows.  Starting in configuration~$\Yt[t]$
at time~$t$, each vertex~$v$ waits for a period of time before reproducing.
This period of time is chosen, independently of other vertices,
according to an exponential distribution with
parameter~$r_{v,\Yt[t]}$.  Therefore, the probability that two
vertices reproduce at once is zero.  Suppose that the first vertex to reproduce
after time~$t$ is vertex~$v$, at time~$t+\tau$.  As in the
discrete-time version of the process, an out-neighbour~$w$ of~$v$ is
chosen u.a.r.\ and the new configuration is given by
$\Yt[t+\tau]=\Yt[t]|_{v\to w}$.

From the definition of the exponential distribution, it is clear that
the probability that a particular vertex~$v$ is the next to reproduce, from
configuration~$\Yt[t]$, is $r_{v,\Yt[t]}/W_{|\Yt[t]|}$. Thus, the
Moran process (as generalised by Lieberman et al.)
is recovered by taking the
sequence of configurations each time a vertex reproduces.\footnote{This is
  closely related to the jump chain, which is defined to be the
  discrete-time chain whose successive states are the states~$\Yt[t]$
  for the successive times~$t$ immediately after the state changes.
  Thus, the jump chain is the chain of ``active''
   steps of the
  discrete-time Moran process (see Section~\ref{sec:active}).}
We can now give the proof of Lemma~\ref{lemma:couple} and
Theorem~\ref{thm:fix}.

\begin{replemma}{lemma:couple}\LEMMAcouple{}\end{replemma}
\begin{proof}
    Suppose that $\Yt[t]\subseteq \Yt'[t]$ for some~$t$.  We couple
    the evolution of $\Yt[t']$ and $\Yt'[t']$ for $t'\geq t$ as
    follows.  For ease of notation, we write $r_{v,t}$
    and~$r'_{\smash{v,t}}$ for $r_{\smash{v,\Yt[t]}}$
    and~$r'_{\smash{v,\Yt'[t]}}$, respectively.  Let
    \begin{equation*}
        S = \{v\in V(G) \mid r_{v,t} < r'_{\smash{v,t}}\} \subseteq \Yt'[t]
    \end{equation*}
    and note that, for $v\in V\setminus S$, $r'_{\smash{v,t}}=r_{v,t}$.
    For $v\in V\!$, let $t_v$~be a random variable drawn from
    $\Exp(r_{v,t})$ and, for $v\in S$, let $t'_v \sim
    \Exp(r'_{\smash{v,t}} - r_{v,t})$.  From the definition of the
    exponential distribution, it is easy to see that, for each $v\in
    S$, $\min(t_v,t'_v) \sim \Exp(r'_{\smash{v,t}})$.

    If some $t_v$ is minimal among $\{t_v\mid v\in V\} \cup \{t'_v\mid
    v\in S\}$, then choose an out-neighbour~$w$ of~$v$ u.a.r.\@ and
    set $\Yt[t+t_v] = \Yt[t]|_{v\to w}$ and $\Yt'[t+t_v] = \Yt'[t]|_{v\to w}$.
    It is clear that $\Yt[t+t_v]\subseteq \Yt'[t+t_v]$.

    Otherwise, some $t'_v$ is minimal.  In this case, set
    $\Yt[t+t'_v]=\Yt[t]$; choose an out-neighbour~$w$ of~$v$ u.a.r.\@
    and set $\Yt'[t+t'_v]=\Yt'[t]|_{v\to w}$.  Since 
    $v\in S\subseteq  \Yt'[t]$, we have
    \begin{equation*}
        \Yt[t+t'_v] = \Yt[t]\subseteq \Yt'[t] \subseteq \Yt'[t+t'_v]\,.
    \end{equation*}

    In both cases, the continuous-time Moran process has been
    faithfully simulated up to time $t+\tau$, 
where $\tau=t_v$ in the first case and $\tau=t'_v$ in the second case,
and the memorylessness of the exponential
    distribution allows the coupling to continue from 
$\Yt[t+\tau]$ and $\Yt'[t+\tau]$.
\end{proof}

The coupling provided in Lemma~\ref{lemma:couple}
could be translated to a coupling for the discrete-time Moran process, though the time steps in the two copies would not be the same since such a coupling was ruled out in Section~\ref{sec:introdom}.
In fact, it will be easy for us to use Lemma~\ref{lemma:couple} directly.

\begin{reptheorem}{thm:fix}\THMfix{}\end{reptheorem}  
\begin{proof}
    We split the proof into two parts: $1\leq r\leq r'$ and $r\leq
    r'\leq 1$.  The remaining case $r\leq 1\leq r'$ follows because
    $f_{G,r}(S)\leq f_{G,1}(S)\leq f_{G,r'}(S')$.
 
    First, suppose that $1\leq r\leq r'\!$.  Let $\Yt[t]$ and
    $\Yt'[t]$ be Moran processes on~$G=(V,E)$ with mutant fitnesses $r$
    and~$r'$, respectively, with $\Yt[0] = S$ and $\Yt'[0] = S'\!$.
    By the coupling lemma, we can couple the processes such that
    $\Yt[t]\subseteq \Yt'[t]$ for all $t\geq 0$.  In particular, if
    there is a~$t$ such that $\Yt[t]=V$, we must have $\Yt'[t]=V$ 
    also.  Therefore, $f_{G,r'}(S')\geq f_{G,r}(S)$. 

    Now, suppose that $r\leq r'\leq 1$.  Observe that the behaviour of
    the Moran process is independent of any consistent scaling of the
    mutant and non-mutant fitnesses, in the following sense.  For any
    $\alpha>0$, the process where mutants have fitness~$r$ and
    non-mutants have fitness~$1$ is identical to the one where they
    have fitness $\alpha r$ and~$\alpha$, respectively.  Let $\Yt[t]$
    be the process where mutants have fitness~$\tfrac{1}{r}\cdot r =
    1$ and non-mutants have fitness~$\tfrac{1}{r}\geq 1$, and let
    $\Yt'[t]$ have mutant fitness~$1$, non-mutant
    fitness~$\tfrac{1}{r'}\geq 1$.  Let $\Yt[0]=S$ and $\Yt'[0]=S'\!$.
    Now, $f_{G,r}(S)$ is the probability that the individuals with
    fitness~$1$ take over the graph in~$\Yt[t]$, which is
    $1-f_{G,1/r}(V\setminus S)$; similarly, 
    $f_{G,r'}(S') = 1 -
    f_{G,1/r'}(V\setminus S')$.  By the first part,
    $f_{G,1/r'}(V\setminus S')\leq f_{G,r}(V\setminus S)$ and the
    result follows.
\end{proof}

Corollaries \ref{cor:monotonicity} (monotonicity)
and~\ref{cor:subsetdom} (subset domination) follow immediately from
Theorem~\ref{thm:fix}, as shown in Section~\ref{sec:introdom}.

\section{Regular digraphs}
\label{sec:regular}
 
This section provides upper and lower bounds on
the absorption time of the Moran process on
regular digraphs. 

\subsection{An upper bound for undirected graphs}

It is clear that the absorption time bounds from~\cite{OurSODA} do
not apply to digraphs. For example, Theorem~7 of~\cite{OurSODA}
gives a polynomial upper bound on the expected absorption time
for all connected undirected graphs, but Theorem~\ref{thm:directed}  shows that 
process takes exponential time on
some strongly connected digraphs.

Since we will be discussing both undirected graphs and digraphs
in this section, we start by  observing that Theorem~7 of~\cite{OurSODA}
can be improved to give an $O(n^3)$ bound in the special case in which the undirected graphs to which it applies
are regular.  
This is certainly not tight (as we shall see below)
but it is a natural place to begin.

\begin{proposition}
\label{prop:oldupper}
The expected absorption time of the Moran process on a connected
$\Delta$-regular 
$n$-vertex undirected graph is at most
$(r/(r-1))n^2 \Delta$.
\end{proposition}
\begin{proof}
Given an undirected graph~$G$ and a
set $S\subseteq V(G)$,
let $\phi(S) = \sum_{v\in S} \tfrac{1}{\deg(v)}$. 
Let $\partial S$ be
the  set of (undirected) edges
between 
vertices in~$S$ and vertices in~$V(G)\setminus S$.
If~$G$ is   $\Delta$-regular  and has $n$ vertices, then $\phi(V(G)) = n/\Delta$.
 The proof of Theorem~7 and Equation~(1) of~\cite{OurSODA} 
 show that the absorption time is at most
 $$ \phi(V(G)) \max \left\{
 \left( 
 \frac{(n+(r-1)|S|)\, \Delta^2}{(r-1)\,|\partial S|}\right)
 \mid \emptyset \subset S \subset V(G)
 \right\}.$$
 The bound follows using $|S|\leq n$ and $|\partial S|\geq 1$. \end{proof}
  
\subsection{Definitions}

We will use the following standard Markov chain definitions.
For more detail, see,   for example,~\cite{Norris}.
We use $(X_t)_{t\geq 0}$ to denote a discrete-time Markov chain~$\calM$
with finite state space~$\Omega$ and transition matrix~$P$.
$T_k = \inf\{ t \geq 1 \mid X_t=k\}$ is the first passage time for visiting
state~$k$ (not counting the initial state~$X_0$).
The  time spent in state~$i$ between visits to state~$k$
is given by 
$$\gamma_i^k =  
\sum_{t=0}^{T_k-1} 1_{X_t=i}, \mbox{ where $X_0=k$.}$$
The chain  is said to be  irreducible if, for every pair of states~$(i,j)$
there is some $t\geq 0$ such that
$\Pr(X_t=j \mid X_0=i)>0$. 
Since $\Omega$ is finite, this implies that
the chain is recurrent, which means
that, for every state $i\in \Omega$,  
$\Pr(\mbox{$X_t=i$ for infinitely many~$t$})=1$. 
We use the following proposition, which (up to minor notational differences)
is the special case of \cite[Theorem 1.7.6]{Norris} corresponding
to finite state spaces.

\begin{proposition} \label{prop:Norristwo}
Let $\calM$ be an irreducible discrete-time Markov chain
with finite state space $\Omega=\{0,\ldots,\omega-1\}$
and transition matrix~$P$.
For $k\in\Omega$, 
let $\lambda=(\lambda_0,\ldots,\lambda_{ \omega-1})$
be a vector of non-negative real numbers with $\lambda_k=1$
satisfying $\lambda P = \lambda$. Then,
for every $j\in \Omega$, $E[\gamma^k_j]=\lambda_j$. \end{proposition}

\subsection{Active steps of the Moran process}
\label{sec:active}
 
We fix $r>1$ and study the Moran process on 
a strongly connected   $\Delta$-regular 
$n$-vertex
digraph $G=(V,E)$ with $n>1$.
We refer to the steps of the process during which the
number of mutants changes as ``active steps''.
As explained in the introduction,
the evolution of the number of mutants, sampled
after each active step, 
corresponds to a one-dimensional random walk 
on $\{0,\ldots,n\}$ which starts at
state~$1$, absorbs at states~$0$ and~$n$, and has
upwards
drift $p=r/(r+1)$. 
 (To see this,
note that the probability that the number of mutants increases
from a size-$k$ state~$S$ is 
$\sum_{e \in E \cap (S\times V(G)\setminus S)} \tfrac{r}{W_k \Delta}$
and the probability that
it decreases is
$\sum_{e \in E \cap (V(G)\setminus S \times S)} \tfrac{1}{W_k \Delta}$
but we showed earlier that the number of edges in each summation is equal
when $G$ is $\Delta$-regular, so the ratio
between these two probabilities is $r$ to~$1$.)

To derive the properties that we need,
we consider a Markov chain $\calM$ with state space $\Omega=\{0,\ldots,n+1\}$.
 The non-zero entries
of the  transition matrix~$P$ of~$\calM$ are as follows.  
$P_{0,n+1}=P_{n,n+1}=1$.
Also, $P_{n+1,1}=1$.
Finally, 
for $1\leq i \leq  n-1$,
$P_{i,i+1}=p$ and $P_{i,i-1}=1-p$.
Starting from state~$1$, the chain simulates the one-dimensional walk discussed above.
State~$n+1$ is a special state of the Markov chain that is visited after an absorbing state of the random walk is reached.
From state~$n+1$, the chain goes back to state~$1$ and repeats the random walk.
We use the following  property of~$\calM$.

\begin{lemma}
\label{lem:calc} 
Let $f= (r^n-r^{n-1})/(r^n-1)$.
Define the vector $\lambda = (\lambda_0,\ldots,\lambda_{n+1})$
as follows.
\begin{align*} 
\lambda_0&=1- f,\\ 
 \lambda_j &= (1+r)(1-f)(r^n-r^j)/(r^n-r),  
 \mbox{ for $1\leq j\leq n-1$,}\\
\lambda_n &= f,\\
\lambda_{n+1} &= 1.
\end{align*} 
Then, for every $j\in \Omega$,
$E[\gamma^{n+1}_j]=\lambda_j$.
 \end{lemma}

\begin{proof} 
Note that $P$ is  irreducible. 
By Proposition~\ref{prop:Norristwo}, 
it suffices to show that  $\lambda P = \lambda$.
First, consider the column vector $P_{*,0}$. This 
is all zero except the entry $P_{1,0}=1-p$
so $\lambda P_{*,0} = (1-p)\lambda_1 = \lambda_0$, as required.
Then  note that $1/(1-f)= r(r^n-1)/(r^n-r)$. So
\begin{align*}
\lambda P_{*,1} &=   (1-p) \lambda_2+\lambda_{n+1}\\  
&=  (1-f)\left( \frac{r^n-r^2}{r^n-r}+ \frac{1}{1-f}\right) \\
& = (1-f)\left(\frac{r^n-r^2 + r(r^n-1)}{r^n-r}\right)\\
&= (1-f)(r+1) =\lambda_1,
\end{align*}
 as required. 
 Next, consider the column vector $P_{*,j}$ for $1<j<n-1$.
In this case, 
\begin{align*}
\lambda P_{*,j} &= 
p \lambda_{j-1} + (1-p) \lambda_{j+1} \\
&= 
(1-f)\left(\frac{r(r^n-r^{j-1}) + (r^n-r^{j+1})}
{r^n-r}\right)\\
&= (1-f) 
\left(\frac{(1+r)(r^n-r^{j})  }
{r^n-r}\right) = \lambda_j,\\
\end{align*}
as required. 
Then  
\begin{align*} \lambda P_{*,n-1} &= p \lambda_{n-2} \\
&= \left(\frac{(1+r)(1-f)}{r^n-r}\right)\left(\frac{r(r^n-r^{n-2})}{r+1}\right) \\
&= \left(\frac{(1+r)(1-f)}{r^n-r}\right)\left(  r^n-r^{n-1} \right)  = \lambda_{n-1},
 \end{align*}
as required.
Furthermore, 
$$\lambda P_{*,n} = p \lambda_{n-1} = 
 (1-f)\left(\frac{r^{n+1}-r^n}{r^n-r}\right).$$
Also, 
$$\frac{1-f}{f}= \frac{r^n-r}{r^{n+1}-r^n},$$
so 
$$\frac{\lambda P_{*,n}}{\lambda_n} = 
\frac{p \lambda_{n-1}}{f} =
\frac{(1-f)}{f} \left(\frac{r^{n+1}-r^n}{r^n-r}\right) = 
\left(\frac{r^n-r}{r^{n+1}-r^n}\right)\left(\frac{r^{n+1}-r^n}{r^n-r}\right)=1,$$
as required. Finally, $\lambda P_{*,n+1} = \lambda_0 + \lambda_n = 1 = \lambda_{n+1}$,
as required. This completes the proof. \end{proof}
 
It is well known \cite{LHN} that~$f$ is the fixation probability 
of the
Moran process on a regular graph.
This is an easy consequence of Lemma~\ref{lem:calc}, but
we don't need it here. We will instead use the following corollary.

\begin{corollary}\label{ourcor}
For all $j\in \{1,\ldots,n-1\}$, 
$  1-\frac{1}{r^2}
\leq E[\gamma_j^{n+1}] \leq 1+\frac1r$.  \end{corollary}  
  
\begin{proof}
From Lemma~\ref{lem:calc},
$$E[\gamma_j^{n+1}] =  \left(\frac{r+1}{r}\right)\left(\frac{r^n-r^j}{r^n-1}\right).$$
The upper bound follows from the fact that $r^n-r^j \leq r^n-1$ (a consequence of $r> 1$ and $j
\geq 0$). 
For the lower bound, note that $E[\gamma_j^{n+1}]$ is minimised at $j=n-1$ 
and 
$$E[\gamma_{n-1}^{n+1}] = 
\left(\frac{r^n}{r^n-1}\right)\left(1-\frac{1}{r^2}\right)
.$$
The lower bound then follows from  $r^n/(r^n-1)\geq 1$.
\end{proof}  
  
\subsection{Absorption time}

Now let the Moran process be 
$(Y_t)_{t\geq 1}$
where each state $Y_t$ is the set $S\subseteq V(G)$
of vertices of~$G$ that are mutants
at the $t$'th step.
The state~$Y_1$ is selected uniformly at random from
the size-$1$ subsets of~$V(G)$.
For each state~$S$,   let
$p(S) = \Pr(Y_{t+1}\neq S \mid Y_t=S)$ and let 
$\mu(S) = \inf\{t\geq 1 \mid Y_{t+1} \neq S, Y_1=S\}$.
$\mu(S)$ is a random variable representing the number of times that the
state~$S$ appears when the process is run, starting from~$S$,
before another state is reached.
It is geometrically distributed with parameter~$p(S)$, so $E[\mu(S)]=1/p(S)$.
The absorption time~$\Tabs$ is
the number of steps needed to
get to state~$0$ or state~$n$, which is
$\Tabs=\inf\{t\geq 1 \mid |Y_t|\in\{0,n\}\}-1$.

Let $\tau_1=1$.  For $j>1$,  let
$\tau_j = \inf\{t>\tau_{j-1} \mid Y_t \neq Y_{t-1}\}$.
The values $\tau_2,\tau_3,\ldots$ are the times at which the state changes.
These are the active steps of the process.
The sequence
$Y_{\tau_1},Y_{\tau_2},\ldots$
is the same as the Moran process except that repeated states are omitted.
Now recall the Markov chain $(X_t)_{t\geq 0}$  with
start state $X_0=n+1$ and recall the definition of the first passage time $T_{n+1}$
which is the first time that the chain returns to state~$n+1$.
Note that
the sequence
$n+1, |Y_{\tau_1}|,|Y_{\tau_2}|,\ldots,|Y_{\tau_{(T_{n+1}-1)}}|$
is a faithful simulation of
the Markov chain 
$X_0,X_1,\ldots,X_{T_{n+1}-1}$ 
starting from state~$X_0=n+1$, 
up until it reaches state~$0$ or state~$n$. 
Also,  
the absorption time  satisfies
$$\Tabs = \tau_{(T_{n+1}-1)} -1 = \sum_{j=2}^{T_{n+1}-1} (\tau_j - \tau_{(j-1)})$$ 
and
for $j\geq 2$,
$\tau_j-\tau_{j-1}$ is distributed as  $\mu(Y_{\tau_{j-1}})$, which
is geometric with parameter $p(Y_{\tau_{j-1}})$.

To  derive upper and lower bounds for $E[\Tabs]$,
we break  the sum into pieces. For $k\in\{1,\ldots,n-1\}$,
let
$$\Tabsj{k} = \sum_{j=2}^{T_{n+1}-1} \Psi_{k,j},$$
where $\Psi_{k,j}$ is geometrically
distributed with parameter $p(Y_{\tau_{j-1}})$
if $|Y_{\tau_{j-1}}|=k$ and $\Psi_{k,j}=0$, otherwise.
Then $\Tabs$ is distributed as $\sum_{k=1}^{n-1} \Tabsj{k}$. 
In order to derive upper and lower bounds, let
$p^+_k = \max\{p(S) \mid |S|=k\}$ and 
$p^-_k = \min\{p(S) \mid |S|=k\}$. Let
$\Tabsj{k}^+= \sum_{j=2}^{T_{n+1}-1} \Psi^+_{k,j}$
where $\Psi^+_{k,j}$ is geometrically
distributed with parameter $p^-_k$
if $|Y_{\tau_{j-1}}|=k$ and $\Psi^+_{k,j}=0$, otherwise.
Let
$\Tabsj{k}^-= \sum_{j=2}^{T_{n+1}-1} \Psi^-_{k,j}$
where $\Psi^-_{k,j}$ is geometrically
distributed with parameter $p^+_k$
if $|Y_{\tau_{j-1}}|=k$ and $\Psi^-_{k,j}=0$, otherwise.
Then  by stochastic domination for the geometric distribution,
\begin{equation}
\label{eq:forWald}
\sum_{k=1}^{n-1} E[\Tabsj{k}^-] \leq E[\Tabs] \leq \sum_{k=1}^{n-1} 
E[\Tabsj{k}^+].
\end{equation}

\begin{theorem}
\label{thm:abs}
$$ 
\left(1-\frac{1}{r^2} \right)
 \sum_{k=1}^{n-1} \frac{1}{p_k^+} \leq E[\Tabs] \leq
\left(1+\frac{1}{r}\right) \sum_{k=1}^{n-1} \frac{1}{p_k^-}.$$
\end{theorem}

\begin{proof} 
 
By~(\ref{eq:forWald}), $E[\Tabs]$ is at most
 $\sum_{k=1}^{n-1} 
E[\Tabsj{k}^+]$. Now $\Tabsj{k}^+$ is a sum of geometric random variables
with parameter~$p^-_k$. The number of random variables in the sum is
$\gamma^{n+1}_k$ which is the number of times that state~$k$
is visited between visits to state~$n+1$ in the Markov chain $(X_i)$.
Since $1/p^-_k$ and $E[\gamma^{n+1}_k]$ are both finite (see Corollary~\ref{ourcor}),
Wald's equality (see  \cite[Theorem 12.3]{MU}) guarantees
that $E[\Tabsj{k}^+] = E[\gamma^{n+1}_k]/p^-_k$.
The upper bound follows from Corollary~\ref{ourcor}. The lower bound is similar.
\end{proof}
 
\subsection{Upper and lower bounds} 
   
We start with the following observation.

\begin{observation}\label{ObsmS}
If $|S|=k$ then
$p(S) = 
\frac{r  \degplus{S}}{W_k \Delta} + 
\frac{\degminus{S}}{W_k \Delta} $
so 
$\frac{1}{p(S)}= \frac{W_k \Delta }{ r \degplus{S} +  \degminus{S}}
$.
\end{observation} 
 
Putting Theorem~\ref{thm:abs} together with Observation~\ref{ObsmS} we get
the following. 
 
 \begin{corollary}\label{cor:together}
 The expected absorption time of the Moran process on a strongly connected $\Delta$-regular
$n$-vertex digraph~$G$ is 
 at least $$  
 \left(1-\frac{1}{r^2}\right)
  W_1 \Delta  
 \sum_{k=1}^{n-1} 
 \frac{1}
 {
 \max  \left\{ 
 r \degplus{S}+\degminus{S} \mid |S|=k \right\}
   }
 $$
   and is at most
 $$ \left(1+\frac1r\right)W_n\Delta \sum_{k=1}^{n-1} 
 \frac{1}
 {
 \min  \left\{ 
 r\degplus{S}+\degminus{S} \mid |S|=k \right\}} .$$
 \end{corollary}

We can now prove Theorem~\ref{thm:regbounds}.

\begin{reptheorem}{thm:regbounds}\THMregbounds{}\end{reptheorem}
\begin{proof}
If $|S|=k$ then we have the trivial bound 
$r \degplus{S} + \degminus{S}   \leq (r+1) k \Delta$,
which, together with Corollary~\ref{cor:together}, establishes
the lower bound.
If a digraph is  strongly connected, then $\degplus{S}$ and $\degminus{S}$ are
at least~$1$ when $1\leq |S|\leq n-1$ so
$r\degplus{S} + \degminus{S} \geq r+1$.
This, together with Corollary~\ref{cor:together}, establishes 
the upper bound.
\end{proof}
 
Note that  the upper bound in Theorem~\ref{thm:regbounds} generalises the one given in
Proposition~\ref{prop:oldupper} to the directed case.
The following observations follow from special cases of 
Corollary~\ref{cor:together}.

 \begin{observation} Suppose that the graph~$G$ is
the undirected clique~$K_n$ (which is $\Delta$-regular with $\Delta=n-1$).
In this case, for $S$ of size~$k$, 
$\degplus{S}= \degminus{S}=k(n-k)$,
so  Corollary~\ref{cor:together} shows that the absorption time
is at most 
$$  n \sum_{k=1}^{n-1} \frac{n-1}{k(n-k)} 
\leq 
  n \sum_{k=1}^{n-1} \frac{n-k}{k(n-k)} +
  n \sum_{k=1}^{n-1} \frac{k}{k(n-k)} 
 \leq 2  n H_{n-1},$$
matching the lower bound from Theorem~\ref{thm:regbounds} up to a constant factor (that 
depends only on~$r$
but not on~$n$).
\end{observation}
 
 \begin{observation} \label{obs:cycle} Suppose that the graph~$G$ is
the undirected cycle~$C_n$ (which is $\Delta$-regular with $\Delta=2$).
Since the process starts with a single mutant, it is easy to see that
the set of mutant vertices must be connected, if it is non-empty.
Therefore, $\degplus{S}=\degminus{S}=2$ for any non-trivial $S$
that is reachable from the initial configuration,
so the absorption time is at least  
$ \left( 1-\frac{1}{r^2}\right) \frac{2n}{r+1} \sum_{k=1}^{n-1} \frac{1}{2} = \Omega(n^2)$, matching
the upper bound from   Theorem~\ref{thm:regbounds} up to a constant factor.
 \end{observation}
 
 \begin{observation} \label{obs:dircycle} Suppose
 that the graph~$G$ is
the directed $n$-vertex cycle (which is $\Delta$-regular with $\Delta=1$).
Again, the mutants remain connected; in this case $ \degplus{S}=\degminus{S}=1$ for any non-trivial $S$
so the absorption time is at least  
$ \left( 1-\frac{1}{r^2}\right) \frac{n}{r+1} \sum_{k=1}^{n-1} \frac{1}{2} = \Omega(n^2)$, matching
the upper bound from Theorem~\ref{thm:regbounds} up to a constant factor. \end{observation}

\subsection{Better upper bounds for
undirected graphs via isoperimetric numbers} 
\label{sec:iso}
 
Suppose that a graph~$G$ is undirected.  
As in the proof of  Proposition~\ref{prop:oldupper}, let 
$\partial S$ be
the  set of (undirected) edges
between 
vertices in~$S$ and vertices in~$V(G)\setminus S$. 
Then $\degplus{S}=\degminus{S}=|\partial S|$.
The  isoperimetric number of the graph~$G$  was defined by Buser~\cite{Buser}
as follows 
$$\iso(G) = \min\left\{
\frac{|\partial S|}{|S|} \mid S \subseteq V(G), 0 < |S| \leq |V(G)|/2
\right\}.$$ 
The quantity $\iso(G)$ is a discrete analogue of the Cheeger isoperimetric
constant.
For graphs with good  expansion, 
Theorem~\ref{thm:isoperimetric} improves
the upper bound in   Theorem~\ref{thm:regbounds}.

\begin{reptheorem}{thm:isoperimetric}\THMisoperimetric{}\end{reptheorem}
\begin{proof}
From Corollary~\ref{cor:together}, the
expected absorption time is at most
 
$$ 
  \frac{\Delta W_n}{r}  \sum_{k=1}^{n-1} 
\frac{1}
{ \min 
 \left\{  {|\partial_S|}
 \mid S \subseteq V(G), |S|=k
 \right\}}.$$
 This is  at most
 
\begin{align*} 
&\frac{\Delta W_n}{r} 
\left(2    \sum_{k=1}^{ \lfloor n/2\rfloor } 
\frac{1}
{ \min 
 \left\{  {|\partial_S|}
 \mid S \subseteq V(G), |S|=k
 \right\}}\right)
 \\
 &  = \frac{\Delta W_n}{r} \left(
2    \sum_{k=1}^{\lfloor n/2 \rfloor} 
\frac{1}
{ k \min 
 \left\{  \frac{|\partial_S|}{k}
 \mid S \subseteq V(G), |S|=k
 \right\}}\right)
 \\
 &\leq \frac{\Delta W_n}{i(G)r}
 \left(
2    \sum_{k=1}^{\lfloor n/2 \rfloor } 
\frac{1}
{ k  }
  \right)
= \frac{2 \Delta W_n H_{\lfloor n/2 \rfloor }}{i(G)r}. \qedhere
  \end{align*}
 \end{proof}

\subsection{Families for which the upper bound is optimal} 
\label{sec:family} 
 
For every fixed $\Delta>2$, we construct an infinite family of
connected, 
$\Delta$-regular undirected graphs for which the upper
bound in   Theorem~\ref{thm:regbounds} is optimal, up to a constant factor
(which may depend upon~$r$ and $\Delta$ but not on~$n$). 

To do this, we define the graph~$H_\Delta$ to be
$K_{\Delta-2,\Delta-1}$ with the addition of edges forming a cycle on
the side with $\Delta-1$ vertices.  Note that $\Delta-2$ vertices have
degree $\Delta-1$ and the others have degree~$\Delta$.

Now, let $G_{\ell,\Delta}$ be the $\Delta$-regular  
graph formed from a cycle $x_1\dots
x_\ell x_1$ and $\ell$~disjoint copies of~$H_\Delta$ by adding an
undirected edge between~$x_i$ and each of the vertices of degree
$\Delta-1$ in the $i$'th copy of~$H_\Delta$, for each $i\in[\ell]$.
Note that $|V(G_{\ell,\Delta})| = 2\ell(\Delta-1)$.

\begin{theorem}
\label{thm:new}
    For $r>1$ and sufficiently large~$\ell$ (with respect to~$r$), the
    expected absorption time of the Moran process on $G_{\ell,\Delta}$
    exceeds $\tfrac{1}{2r}(1-\tfrac{1}{r})\ell^2\!$.
\end{theorem}
\begin{proof}
    Let $n=n(\ell,\Delta) = |V(G_{\ell,\Delta})|$ and let
    $\Tdisc=\ell^2/r$.

    Let $(Y_t)_{t\geq 1}$ be the discrete-time Moran process on~$G_{\ell,\Delta}$
    and consider the following events.  Let $\Fixation^*$~be the event
    that $Y_{\Tdisc+1} = V(G_{\ell,\Delta})$, i.e., that the process reaches
    fixation in at most $\Tdisc$~steps.  Let $\Extinction^*$~be the event
    that $Y_{\Tdisc+1} = \emptyset$,
    i.e., that the process reaches extinction in at most $\Tdisc$~steps.
    $\Pr(\Extinction^*)$~is at most the extinction probability, which
    is less than~$\tfrac{1}{r}$
since $G$ is regular, so   
the fixation probability is the quantity~$f$
from Lemma~\ref{lem:calc} \cite{LHN} which  exceeds $1-\tfrac1r$ for $r>1$.

    To bound $\Pr(\Fixation^*)$, consider the continuous-time version
    of the process, $\Yt[t]$.  We will show that, by time~$\Tcont =
    2\Tdisc/n = 2\ell^2/rn$, it is very likely that the continuous
    process will have had at least~$\Tdisc$ reproductions.  Let
    $\Slow$ (for ``slow'') be the event that the continuous process
    has not had $\Tdisc$~reproductions by time~$\Tcont$. Let
    $\tFixation^*$ be the event that it has
    reached fixation by time~$\Tcont$.  We have
    \begin{equation*}
        \Pr(\Fixation^*) 
       =
       \Pr(\Fixation^* \wedge \neg \Slow) + \Pr(\Fixation^* \wedge \Slow) 
        \leq \Pr(\tFixation^*) + \Pr(\Slow).
    \end{equation*}

    In the continuous process, 
each of the $n$~vertices reproduces at rate at least~$1$
    so the number~$N$ of reproductions up to time~$\Tcont$ is stochastically
    bounded below by a Poisson random variable with
parameter~$n\Tcont=2\ell^2/r$. 
By a Chernoff-type argument
    \cite[Theorem~5.4]{MU}, we have
    \begin{equation*}
        \Pr(\Slow) 
        = 
                 \Pr(N \leq \Tdisc)
            \leq e^{-n\Tcont} \left(\frac{en\Tcont}{\Tdisc}\right)^{\Tdisc}
            = e^{-2\ell^2/r} (2e)^{\ell^2/r}
            < \tfrac{1}{4}\left(1-\tfrac{1}{r}\right)\,,
    \end{equation*}
    for large enough~$\ell$.

    To bound $\Pr(\tFixation^*)$, consider the process $\Zt[t]$
    on~$G_{\ell,\Delta}$ that behaves like $\Yt[t]$ except for the two
    following points.
    \begin{itemize}
    \item For some~$i$, we have $\Yt[0]=\{x_i\}$ or $\Yt[0]$~is in the
        copy of $H_\Delta$ attached to~$x_i$.  Let $\Zt[0] = \Yt[0]\cup
        \{x_i\}$.
    \item No mutant in the cycle $x_1\dots x_\ell x_1$ can ever be
        replaced by a non-mutant.  That is, if, at time~$t$, a
        non-mutant neighbour of some~$x_i$ ($i\in [\ell]$) is selected
        to reproduce to~$x_i$, then the state does not change.
    \end{itemize}

    We couple the processes $\Yt[t]$ and $\Zt[t]$ as follows.  Let
    $t$~be such that $\Yt[t]\subseteq \Zt[t]$, noting that $t=0$ has
    this property.  The coupling lemma (Lemma~\ref{lemma:couple})
    allows us to maintain $\Yt[t+\tau]\subseteq \Zt[t+\tau]$ until the
    next time, $t'\!$, at which a mutant at one of the $x_i$~is
    replaced by a non-mutant in~$\Yt[t]$.  This maintains the property
    that $\Yt[t']\subseteq \Zt[t']$, so the coupling can be restarted from
    this point.

Now $\Pr(\tFixation^*)  = \Pr(\Yt[\Tcont]=V(G_{\ell,\Delta}))
\leq \Pr(\Zt[\Tcont]=V(G_{\ell,\Delta}))$, so we will find an upper bound
for   $\Pr(\Zt[\Tcont]=V(G_{\ell,\Delta}))$.
    Let $C = \{x_1, \dots, x_\ell\}$.  The set $\Zt[t]\cap C$ is
    non-empty, non-decreasing and connected in~$G_{\ell,\Delta}$.  
If $\Zt[t]\cap C$ is a proper subset of~$C$ then it
    increases exactly when one of the two mutants in~$C$ reproduces to
    its non-mutant neighbour in~$C$ or, if there is only one mutant
    in~$C$, when that mutant reproduces to either of its neighbours in
    the cycle.  In both cases, this happens with rate
    $\tfrac{2r}{\Delta}$, so $|\Zt[t]\cap C|$ 
 is bounded from above
by a Poisson random variable   with parameter $\tfrac{2r}{\Delta}t
$.  
Let $\lambda^*=4\ell/9$.
For $t=\Tcont$ 
the parameter is 
$\tfrac{2r}{\Delta}\Tcont = \tfrac{4 \ell^2}{n \Delta }
  < 
            \frac{4\ell}{\Delta^2}
                         \leq \lambda^*
$
so $|\Zt[t]\cap C|$ is bounded above by
a Poisson random variable~$\Psi^*$  
with parameter~$\lambda^*$.
Therefore,
 $        E[|\Zt[\Tcont] \cap C|] \leq \lambda^*$,
    and we have
    \begin{align*}
        \Pr\left(|\Zt[\Tcont]\cap C| \geq \tfrac{8}{9}\ell\right)
\leq \Pr(\Psi^* \geq 2 \lambda^*) \leq 
\left(\tfrac{e}{4}\right)^{\lambda^*}
            < \tfrac{1}{4}\left(1-\tfrac{1}{r}\right)\,,
    \end{align*}
    for large enough~$\ell$.
    Now,
    \begin{equation*}
        E[\Tabs] \geq E\left[\Tabs\mid
                          \overline{\Fixation^* \cup \Extinction^*}\,
                            \right]
                      \times \Pr\left(\overline{\Fixation^* \cup
                                     \Extinction^*}\right)\,.
    \end{equation*}
    Clearly, the expected absorption time of the discrete process
    conditioned on absorption not occurring 
    within~$\ell^2/r$ steps is at least $\ell^2/r$.  
    Meanwhile,
    \begin{equation*}
        \Pr\left(\overline{\Fixation^* \cup \Extinction^*}\right)
           > 1 - \tfrac{1}{r} - 2\tfrac{1}{4}\left(1-\tfrac{1}{r}\right)
           = \tfrac{1}{2}\left(1-\tfrac{1}{r}\right)
    \end{equation*}
    for large enough~$\ell$ (with respect to~$r$).
\end{proof}

Thus, we have shown that the $O(n^2)$ upper bound of
Theorem~\ref{thm:regbounds} is tight up to a constant factor (which may
depend on $r$ and~$\Delta$, but not on~$n$).

\begin{reptheorem}{thm:tight}\THMtight{}\end{reptheorem}
\begin{proof}
    For a given value of~$r$, Let $\ell_r$ be the smallest value
    of~$\ell$ for which that Theorem~\ref{thm:new} applies.  Take
    $\mathcal{G}=\{ G_{\ell,\Delta} \mid \ell \geq \ell_r \}$ and the
    result is immediate from Theorem~\ref{thm:new} and the fact that
    $G_{\ell,\Delta}$ has $2\ell(\Delta-1)$ vertices.
\end{proof}

\section{General digraphs} 
\label{sec:directed}

Fix $r>1$ and let $\epsilon_r = \min(r-1,\,1)$.
Theorem~7 of \cite{OurSODA} shows
that the expected absorption time of the Moran process on a 
connected $n$-vertex undirected graph 
is at most $(r/(r-1)) n^4$. 
Theorem~\ref{thm:regbounds}  
shows that
the expected absorption time on
a   strongly connected
$\Delta$-regular digraph is at most $n^2 \Delta$.
In contrast, we show that  there is an infinite family 
of strongly connected digraphs  such that the
expected absorption time of the Moran process on an $n$-vertex graph 
from the family  is $2^{\Omega(n)}\!$.

\subsection{The family of graphs}

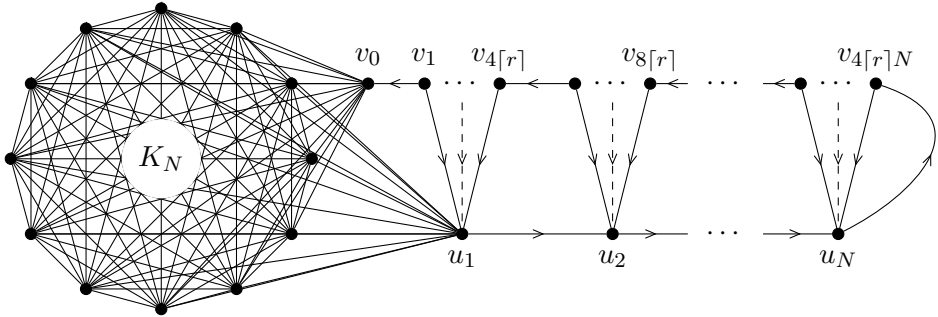
\begin{figure}[t]
\begin{center}
\begin{tikzpicture}
    \tikzstyle{vertex}=[fill=black, draw=black, circle, inner sep=1.5pt]

    \foreach \x in {0,...,11} \node[vertex] (\x) at (30*\x:2) {};

    \foreach \x in {0,...,10}
        \foreach \y in {\x,...,11}
            \draw (\x) -- (\y);

    \node[vertex] (u1) at (4,-1)   [label=-90:{$u_1$}] {};
    \node[vertex] (u2) at (6,-1)   [label=-90:{$u_2$}] {};
    \node              at (7.5,-1) {$\cdots$};
    \node[vertex] (uN) at (9,-1)   [label=-90:{$u_N$}] {};

    \begin{scope}[decoration={markings, mark=at position 0.55 with {\arrow{angle 45}}}]
    \draw[postaction={decorate}] (u1)--(u2);
    \draw[postaction={decorate}] (u2)--(7,-1);
    \draw[postaction={decorate}] (8,-1)--(uN);
    \end{scope}

    \node[vertex] (v0) at (2.75,1) [label={[label distance=2.5pt]90:{$\smash{v_0}$}}] {};

    \node[vertex] (v1) at (3.5,1) [label={[label distance=2.5pt]90:{$\smash{v_1}$}}] {};
    \node              at (4,1) {$\cdots$};
    \node[vertex] (v2) at (4.5,1) [label={[label distance=2.5pt]90:{$\smash{v_{4\ceil{r}}}$}}]{};

    \node[vertex] (v3) at (5.5,1) {};
    \node              at (6,1) {$\cdots$};
    \node[vertex] (v4) at (6.5,1) [label={[label distance=2.5pt]90:{$\smash{v_{8\ceil{r}}}$}}] {};

    \node              at (7.5,1) {$\cdots$};

    \node[vertex] (v5) at (8.5,1) {};
    \node              at (9,1) {$\cdots$};
    \node[vertex] (v6) at (9.5,1) [label={[label distance=2.5pt]90:{$\smash{v_{4\ceil{r}N}}$}}] {};

    \begin{scope}[decoration={markings, mark=at position 0.67 with {\arrow{angle 45}}}]
    \draw[postaction=decorate] (v1)--(v0);
    \draw[postaction=decorate] (v3)--(v2);
    \draw[postaction=decorate] (7,1)--(v4);
    \draw[postaction=decorate] (v5)--(8,1);
    \end{scope}

    \foreach \x in {0, ..., 11}
        \draw (\x)--(u1);

    \foreach \x in {0, ..., 11}
        \draw (v0)--(\x);

    \begin{scope}[decoration={markings, mark=at position 0.55 with {\arrow{angle 45}}}]
    \draw[postaction=decorate] (v1)--(u1);
    \draw[postaction=decorate] (v2)--(u1);

    \draw[postaction=decorate] (v3)--(u2);
    \draw[postaction=decorate] (v4)--(u2);

    \draw[postaction=decorate] (v5)--(uN);
    \draw[postaction=decorate] (v6)--(uN);

    \draw[postaction=decorate] (uN) .. controls (10,-0.5) and (11,0.5) .. (v6);
    \end{scope}

    \begin{scope}[decoration={markings, mark=at position 0.5 with {\arrow{angle 45}}}]
    \draw[postaction=decorate, dashed] (4,0.75)--(u1);
    \draw[postaction=decorate, dashed] (6,0.75)--(u2);
    \draw[postaction=decorate, dashed] (9,0.75)--(uN);
    \end{scope}

    \node[draw=white, fill=white, circle] at (0,0) {$K_N$};
\end{tikzpicture}
\caption{The graph $G_{r,N}$.  The edges within the clique are
   bidirectional; $v_0$ sends a directed edge to every vertex in the
  clique and $u_1$ receives one from each.  Other edges are directed
  as indicated.\label{fig:GrN}}
\end{center}
\end{figure}

Let $G_{r,N}$ be the disjoint union of the complete graph~$K_N$ (with bidirectional edges), a
directed path $P=u_1\dots u_N$ and a directed path $Q=v_{4\ceil{r}N}
\dots v_0$, along with the directed edge $(u_N,v_{4\ceil{r}N})$ and the following
directed edges (see Figure~\ref{fig:GrN}):
\begin{itemize}
\item $(x,u_1)$ and $(v_0,x)$ for every $x\in K_N$;
\item $(v_{4\ceil{r}(i-1)+j},u_i)$ for each $i\in[N]$, $j\in[4\ceil{r}]$.
\end{itemize}

The intuition is as follows.   Consider the Moran process on~$G_{r,N}$.
With probability close to~$\tfrac{1}{4\ceil{r}+2}$,   the initial mutant is in the clique
(Observation~\ref{obs:StartClique}).  
Conditioned on this, Lemmas~\ref{lem:SometimePlentyClique},
 \ref{lem:Extinction}
and~\ref{lem:Fixation} allow us to show that
it is fairly likely that 
there is a time during the first~$N^3$ steps when  the clique is
half full, but that
absorption does not happen in the first~$T^*(N)$ steps 
for a function~$T^*$ which is exponential in~$N$.
Of course, the expected absorption time conditioned on this is at least~$T^*(N)$,
so  we conclude (Theorem~\ref{thm:directed}) that the 
overall expected absorption time is at least~$T^*(N)$.
  
The  main challenge of the proof is the second step --- showing that
is   it is fairly likely that 
there is a time during the first~$N^3$ steps when the clique is half full, but that
absorption does not happen in the first~$T^*(N)$ steps.
The fact that the clique becomes half full
 (Lemma~\ref{lem:SometimePlentyClique}) follows by dominating the number
of mutants in the clique during the initial stages of the process by an
appropriate one-dimensional random walk, and then showing that
sufficiently many random-walk steps are actually taken during the 
first~$N^3$ steps of the process.
The fact that extinction is then unlikely in the first~$T^*(N)$ steps  
(Lemma~\ref{lem:Extinction}) follows
from the fact that the many mutants in the clique are unlikely to become
extinct very quickly. On the other hand, the fact that fixation is unlikely in the
first~$T^*(N)$ steps  (Lemma~\ref{lem:Fixation}) follows from the fact that
mutants make slow progress along the path~$P$ because 
because vertices in~$Q$ tend to push the ``mutant frontier'' backwards
towards~$u_1$. However, the chain of mutants has to
push all the way around this chain in order for fixation to occur.

\subsection{The one-dimensional random walk}

The following Lemma is Example~3.9.6 from~\cite{GS2001:Probability}.

\begin{lemma}
\label{lemma:absorb}
    Let $(Z_t)_{t\geq 0}$ be the random walk on $\{0, \dots, n\}$ with
    absorbing barriers at $0$ and~$n$ and, for $0<Z_t<n$, let $Z_{t+1} =
    Z_t+1$ with probability~$p\neq\tfrac12$ 
and $Z_{t+1}=Z_{t}-1$ with probability
    $q=1-p$.  Let $p_i$ be the probability of absorption at~$0$, given
    that $Z_0=i$.  Writing $\rho = q/p$,
    \begin{equation*}
        p_i = \frac{\rho^i - \rho^n}{1 - \rho^n}\,.
    \end{equation*}
\end{lemma}
 
\subsection{Bounding the absorption time}

Consider the Moran process 
$(Y_t)_{t\geq 1}$
on $G_{r,N}$.
We will assume that $N$ is sufficiently large with respect to~$r$ ---
the exact inequalities that we need will be presented as they arise in the proof.
Let $n = 1+(4\ceil{r}+2)N$ be the number of vertices of $G_{r,N}$.
Let $W_t = n + (r-1)|Y_t|$ be the total fitness of~$Y_t$.
Let~$\Tabs$ be the absorption time of the process.
Our goal is to show that $E[\Tabs]$ is exponentially large, as a function of~$N$.
Let  
$$T^*(N) = \floor{\left(\frac{\epsilon_r}{32}\right)(2^N-1)}$$
and let $N'$ denote $\floor{N/2}$.
We now identify various events which we will study in the lemmas
that follow.
\begin{itemize}
\item
Let $\StartClique$ be the event that $Y_1 \subseteq K_N$
(mnemonic: $\StartClique$ is the event that the initial 
mutant starts in the clique; $\mathcal S$ is for ``Starts'' and $\mathcal C$ is for ``Clique'').
\item Let $\PlentyClique{t}$ be the event
that $|Y_t \cap K_N| \geq N'$ (mnemonic: $\PlentyClique{t}$ is the event that
the clique is half full at time~$t$. $\mathcal H$ is for ``Half'').
\item Let $\SometimePlentyClique = \bigcup_{t\in [N^3+1]} \PlentyClique{t}$. 
\item Let $\Fixation^*$ be the event that $Y_{T^*(N)+1} = V(G_{r,N})$.
(mnemonic: $\Fixation^*$ is the event that fixation occurs   after at most $T^*(N)$ steps;
$\mathcal F$ is for ``Fixation'').
\item Let $\Extinction^*$ be the event that 
$Y_{T^*(N)+1}=\emptyset$ (mnemonic: $\Extinction^*$ is the event that
extinction occurs after at most $T^*(N)$ steps; $\mathcal{E}$ is for ``Extinction'').

\end{itemize}

\begin{observation}
\label{obs:StartClique}
$$\Pr(\StartClique)  =\frac{N}{n}= \frac{N}{N(4 \ceil{r}+2)+1}.$$
\end{observation}

\begin{lemma}
\label{lem:SometimePlentyClique}
$\Pr(\SometimePlentyClique \mid \StartClique) \geq  \epsilon_r/8.$
\end{lemma}

\begin{proof} 

Let $Z'_t = |Y_t \cap K_N|$.  
We will condition on the fact that
$\StartClique$ occurs, so $Z'_1=1$.
If $Z'_t \in \{1,\ldots,N'-1\}$ then  
for all $z_t$ and $w_t$,
$$\Pr(Z'_{t+1}=z'_t+1\mid Z'_t=z_t, W_t=w_t)\geq \left(\frac{r z'_t}{w_t}\right)\left( \frac{N-z'_t}{N}\right).$$
(The probability is greater than this 
if $v_0$ is in~$Y_t$.)
Also,
$$\Pr(Z'_{t+1}=z'_t-1\mid Z'_t=z'_t,W_t=w_t) \leq 
       \left( \frac{1}{w_t}\right)\left(\frac{z'_t}{N} \right)+ \left(\frac{N-z'_t}{w_t}\right)\left(\frac{z'_t}{N}
       \right),$$       
where the first term comes from reproduction from a non-mutant  
at~$v_0$ and the second from reproduction within the clique. 
Also, $Z'_{t+1} \in \{Z'_t-1,Z'_t,Z'_t+1'\}$.
Now let
$$p' = \frac{r}{r+1+\tfrac{2}{N}} $$
and note that
$$\frac{\Pr(Z'_{t+1}=z'_t+1 \mid Z'_t=z'_t)}
{\Pr(Z'_{t+1}=z'_t+1\mid Z'_t=z'_t) + \Pr(Z'_{t+1}=z'_t-1\mid Z'_t=z'_t)} 
\geq
\frac{r(N-z'_t)}{r(N-z'_t)+1+(N-z'_t)}
\geq  
p'.$$ 
The restriction of~$(Z'_t)$ to steps where the state changes,
stopping when $Z'_t$ reaches~$0$ or~$N'$
is a process which
is dominated below by $Z_t$,   a random walk on $\{0,\ldots,N'\}$
that starts at~$1$ and absorbs at~$0$ and~$N'$ and has  parameter~$p'$.

Now let $\EmptyClique$ be the event that
there  is a $t$ with $Y_t\cap K_N=\emptyset$
 such that, for all $ t'<t$, we have $|Y_{t'}\cap K_N| < N'$.
$\EmptyClique$ is the event that the clique becomes empty before
it becomes half full.
Then applying Lemma~\ref{lemma:absorb} to the dominating
random walk,
$\Pr(\EmptyClique \mid \StartClique) \leq 
\frac{\rho - \rho^{N'}}{1 - \rho^{N'}}\leq \rho$, where 
$\rho=(1-p')/p' \leq (N+2)/(N+\epsilon_r N)$.
Since  we are taking~$N$ to be sufficiently large with respect
to $\epsilon_r$ (in particular, we will take
$N\geq 4/\epsilon_r$)
we have
$\rho \leq (1+\epsilon_r/2)/(1+\epsilon_r)$ which is at most $1-\epsilon_r/4$
since $\epsilon_r \leq 1$. 

Let $\Quick$ be the event
that there is a $t\in[N^3+1]$ with $|Y_t\cap K_n| \notin \{1,\ldots,N'-1\}$.
($\Quick$ is the event that the size of the clique
changes quickly --- it takes at most~$N^3$ steps to either become empty
or to become at least half full). 
 We will show below that
$\Pr(\Quick \mid \StartClique) \geq 1- \epsilon_r/8$.
Note that if $\Quick$ occurs but $\EmptyClique$ does not occur then
$\SometimePlentyClique$ occurs.
So
$$\Pr(\SometimePlentyClique \mid \StartClique) \geq
\Pr(\Quick \setminus \EmptyClique \mid \StartClique) \geq 
\Pr(\Quick \mid \StartClique) - \Pr(\EmptyClique \mid \StartClique) \geq
(1- \epsilon_r/8)-(1-\epsilon_r/4)\geq \epsilon_r/8,$$
which would complete the proof.

We conclude the proof, then, by showing
$\Pr(\neg \Quick \mid \StartClique) \leq \epsilon_r/8$.
So we will show
that
$\Pr(Z'_1,\ldots,Z'_{N^3+1} \in \{1,\ldots,N'-1\}  \mid Z'_1=1) \leq \epsilon_r/8$.
For this,
let $\delta=\epsilon_r/20$ and let $\Psi = N'/\delta$.
We require that $N$ is sufficiently large 
with respect to~$r$, so
$N^3 \geq 4 n \Psi$.

Let $\Upsilon_1,\ldots,\Upsilon_{1+\Psi}$ be 
$\Psi$ steps of a random walk on the integers that starts
with $\Upsilon_1=1$ and has $\Upsilon_{t+1}=\Upsilon_t+1$ with probability~$p'$
and $\Upsilon_{t+1}=\Upsilon_t-1$ with probability~$1-p'$.
It is likely that
the state $\Upsilon_t$ increases at least
$(1-\delta) p' \Psi$ times.
By a Chernoff bound (e.g., \cite[Theorem 4.5]{MU}), the probability that this does not happen
is at most $\exp(-p'\Psi \delta^2/2)$.
Since $N$ is sufficiently large with respect to~$r$ (and therefore $N'$ is
sufficiently large with respect to~$r$)
and $p'\geq r/(r+2)$,
this probability is at most $\epsilon_r/16$.
(This calculation is not tight in any way --- $\epsilon_r/16$ happens
to be sufficient.)
If there are at least $(1-\delta)p' \Psi$ increases
then 
$$\Upsilon_{\Psi+1} - \Upsilon_1 \geq (2(1-\delta)p'-1)\Psi.$$
But since 
$$\delta = \frac{\left(\frac{\epsilon_r}{4}\right)}{5}< \frac{\frac{\epsilon_r}{2}-\frac{1}{N}}{2+\frac{3 \epsilon_r}{2}+\frac{1}{N}},$$
we have 
$$p' = \frac{1+\epsilon_r}{2(1+\frac{\epsilon_r}{2}+\frac{1}{N})} > \frac{1+\delta}{2(1-\delta)}.$$
so 
\begin{equation}
\label{eq:driftup}
\Upsilon_{\Psi+1} - \Upsilon_1 \geq (2(1-\delta)p'-1)\Psi \geq \delta \Psi = N'.
\end{equation}

Now if $Z'_t \in \{1,\ldots,N'-1\}$ then  
$$\Pr(Z'_{t+1}\neq z'_t \mid Z'_t=z'_t,W_t=w_t) 
\geq \frac{(r+1)z'_t (N-z'_t)}{N w_t} \geq \frac{1}{2n}.$$
(Again, this calculation is not tight, but $1/(2n)$ suffices.)
If we select $N^3$ Bernoulli random variables, each with success probability $\frac{1}{2n}$,
then, by another Chernoff bound, the probability that we fail to get at least 
$N^3/(4n)\geq \Psi$ successes is  
at most $\exp(- N^3/(16n)) \leq \epsilon_r/16$.

Now $\Pr(Z'_1,\ldots,Z'_{N^3+1} \in \{1,\ldots,N'-1\}  \mid Z'_1=1)$
is at most the 
sum of two probabilities.
\begin{itemize}
\item The probability that
$Z'_1,\ldots,Z'_{N^3+1}$ are all in $\{1,\ldots,N'-1\}$ and
there fewer than~$\Psi$ values~$t$ with
$Z'_{t+1}\neq Z'_t$. This is dominated by the the
selection of $N^3$ Bernoulli random variables as above, and the probability
is at most $\epsilon_r/16$.
\item The probability that
$Z'_1,\ldots,Z'_{N^3+1}$ are all in $\{1,\ldots,N'-1\}$ and
there  are at least $\Psi$ values~$t$ with
$Z'_{t+1}\neq Z'_t$.
In order to bound this probability, imagine the
evolution of the process proceeding in two sub-steps at each step.
First, decide whether $Z'_{t+1}\neq Z'_t$ with the appropriate probability.
If so, select $Z'_{t+1} \in \{Z'_t+1,Z'_t-1\}$ with the appropriate probability.
The probability of the whole event is then dominated above by the probability that
$\Upsilon_{\Psi+1} - \Upsilon_1 < N'$, which is at most $\epsilon_r/16$, as
we showed above. (If this difference
is at least $N'$ then either the $Z'_t$ process hits $0$ before it changes
for the $\Psi$'th time, or
$Z'_t$ reaches $N'$.)\qedhere
\end{itemize}
\end{proof}

\begin{lemma}
\label{lem:PlentyClique}
$\Pr(Y_{T^*(N)+t}=\emptyset \mid \PlentyClique{t}\wedge \StartClique) \leq 
\epsilon_r/(32 (N^3+1)).$
\end{lemma}

\begin{proof}

As in the proof of Lemma~\ref{lem:SometimePlentyClique},
let $Z'_t = |Y_t \cap K_N|$.  
If $\PlentyClique{t}$ holds then $Z'_t \geq N'$.
Let $\Gamma$ be the number of distinct values $t'$
with $t < t' < \inf\{t''>t \mid Z'_{t''}=0\}$
satisfying $Y_{t'}=N'$.
We will show that $\Pr(\Gamma < T^*(N)) \leq \epsilon_r/(32 (N^3+1))$.  

For any $T>t$, suppose that $Z'_T=N'-1$ and let
$T' = \inf\{t''>t \mid Z'_{t''}\in\{0,N'\} \}$.
Let $\pi = \Pr(Z'_{T'}=0)$.
By the argument in the proof of Lemma~\ref{lem:SometimePlentyClique},
$\pi$ is at most the probability that a random walk $Z_t$
on $\{0,\ldots,N'\}$ which starts at $N'-1$ and absorbs at~$0$ and~$N'$ and
has parameter~$p'$ absorbs at~$0$.
By Lemma~\ref{lemma:absorb},
$$\pi \leq \frac{\rho^{N'-1}-\rho^N}{1-\rho^N} \leq \rho^{N'-1},$$
where $\rho=(N+2)/(N+\epsilon_r N)<1$.
Then $\Pr(\Gamma < T^*(N)) \leq  T^*(N) \rho^{N'-1}$.
We will choose
$N$ to be sufficiently large 
so that
$$T^*(N)= \floor{\left(\frac{\epsilon_r}{32}\right)(2^N-1)} \leq  
\left(\frac{\epsilon_r}{32(N^3+1)}\right){\left(\frac{N+\epsilon_r N}{N+2}\right)}^{N'-1}.$$
Then 
$\Pr(\Gamma < T^*(N)) \leq T^*(N) \rho^{N'-1}  \leq \epsilon_r/(32 (N^3+1))$, which completes the proof. \end{proof}

\begin{lemma}
\label{lem:Extinction}
$\Pr(\Extinction^* \mid \SometimePlentyClique \wedge \StartClique) \leq \epsilon_r/32.$
\end{lemma}

\begin{proof}

This follows easily from Lemma~\ref{lem:PlentyClique}
using the following summation.
\begin{align*}
\Pr(\Extinction^* \mid \SometimePlentyClique \wedge \StartClique) 
&= \frac{\Pr(\Extinction^* \wedge \SometimePlentyClique \wedge \StartClique)}
{\Pr(\SometimePlentyClique \wedge \StartClique)} \leq
\frac
{\sum_{t\in [N^3+1]} \Pr(\Extinction^* \wedge \PlentyClique{t} \wedge \StartClique)}
{\Pr(\SometimePlentyClique \wedge \StartClique)} \\
&=
\sum_{t\in[N^3+1]} \frac
{\Pr(\Extinction^* \mid \PlentyClique{t} \wedge \StartClique)\Pr(\PlentyClique{t} \mid \StartClique) \Pr(\StartClique)}
{\Pr(\SometimePlentyClique \mid \StartClique) \Pr(\StartClique)}
\\
&\leq 
\sum_{t\in[N^3+1]} 
\Pr(\Extinction^* \mid \PlentyClique{t} \wedge \StartClique)
\\
& \leq
\sum_{t\in[N^3+1]} 
\Pr( Y_{T^*(N)+t}=\emptyset \mid \PlentyClique{t} \wedge \StartClique).\qedhere
\end{align*}
 \end{proof}

\begin{lemma}
\label{lem:Fixation}
$\Pr(\Fixation^* \mid \StartClique) \leq \epsilon_r/32.$
\end{lemma}

\begin{proof}

Recall the paths $P=u_1\dots u_N$ and~$Q=v_{4\ceil{r}N} \cdots v_0 $ in~$G_{r,N}$. Define $U_t$ 
as follows. 
\begin{itemize}
\item If $Y_t \cap Q$ is non-empty then $U_t=N$.
\item If $Y_t \cap Q$ and $Y_t \cap P$ are both empty then $U_t=0$.
\item If $Y_t \cap Q$ is empty and $Y_t \cap P$ is non-empty
then $U_t = \max\{i \mid u_i \in Y_t\}$.
\end{itemize}
Let $\tau = \inf\{t \mid U_t=N\}$.
We will show that $\Pr(\tau < T^*(N))\leq \epsilon_r/32$.
If $U_t \in \{1,\ldots,N-1\}$ then 
$\Pr(U_{t+1}=U_t+1) = \frac{r}{W_t}$ 
and $\Pr(U_{t+1} = U_t-1) \geq \left(\frac{4\ceil{r}}{W_t}\right)
\left(\frac12\right)\geq \frac{2r}{W_t}$.
Also, $U_{t+1} \in \{U_{t-1},U_t,U_t+1\}$.

Let $F_t = \inf\{t'>t \mid U_{t'}\in\{0,N\}\}$
and $\gamma = \Pr(U_{F_t}=N \mid U_t=1)$.
$\gamma$ is at most the probability of absorbing at~$N$
in a random walk on $\{0,\ldots,N\}$ that starts at~$1$,
absorbs at~$0$ and~$N$ and has parameter~$1/3$ (twice as likely to
go down as to go up). By Lemma~\ref{lemma:absorb}
$$\gamma \leq 1-\left(\frac{ \rho-\rho^N}{1-\rho^N}\right),$$ where $\rho=2$,
so $\gamma \leq1/(2^N-1)$. 
Now let $\Psi$ be the number of times~$t<\tau$
with $U_t=0$. 
Then $\Psi \leq \tau$ so
\begin{equation*}
\Pr(\tau < T^*(N)) \leq \Pr (\Psi < T^*(N)) \leq T^*(N) \gamma \leq \epsilon_r/32.
\qedhere
\end{equation*}
\end{proof}

Putting together the lemmas in this section, we 
prove \textbf{Theorem~\ref{thm:directed}}.
Recall that $T^*(N) = \floor{\left(\frac{\epsilon_r}{32}\right)(2^N-1)}$.
For convenience, we restate the theorem using this notation.

\begin{theorem}
\label{thm:newdirected}
Fix $r>1$ and let $\epsilon_r = \min(r-1,\,1)$.
  Suppose that~$N$ is sufficiently large with respect to~$r$
and
consider the Moran process 
$(Y_t)_{t\geq 1}$
on $G_{r,N}$.
Let~$\Tabs$ be the absorption time of the process.
Then $E[\Tabs] > \tfrac{1}{16}T^*(N)\,\epsilon_r/(4\ceil{r}+3)$. \end{theorem}
\begin{proof}
$$ 
E[\Tabs] \geq E[\Tabs \mid \StartClique \wedge \SometimePlentyClique \setminus
(\Fixation^* \cup \Extinction^*)] \times \Pr(\SometimePlentyClique
\setminus (\Fixation^* \cup \Extinction^*)\mid \StartClique) \times \Pr(\StartClique).
$$
 
The first term on the right-hand side is greater than~$T^*(N)$ by the
definition of the excluded events $\Fixation^*$ and $\Extinction^*$.
Lemmas~\ref{lem:SometimePlentyClique},
 \ref{lem:Extinction}
and~\ref{lem:Fixation}
show that the second term on the right-hand side
is at least $ \epsilon_r/8 - 2 \epsilon_r/32\leq \epsilon_r/16$.
Finally, Observation~\ref{obs:StartClique} shows that the third term on the
right-hand side is at least  
$ 1/(4\ceil{r}+3)$.
\end{proof}

\bibliographystyle{plain}
\bibliography{\jobname}

\begin{thebibliography}{10}

\bibitem{AS2006:fixtime}
Tibor Antal and Istv{\'a}n Scheuring.
\newblock Fixation of strategies for an evolutionary game in finite
  populations.
\newblock {\em Bulletin of Mathematical Biology}, 68(8):1923--1944, 2006.

\bibitem{ARLV2001:Influence}
Chalee Asavathiratham, Sandip Roy, Bernard Lesieutre, and George Verghese.
\newblock The influence model.
\newblock {\em IEEE Control Systems}, 21(6):52--64, 2001.

\bibitem{Ber2001:Monopolies}
Eli Berge.
\newblock Dynamic monopolies of constant size.
\newblock {\em Journal of Combinatorial Theory, Series B}, 83(2):191--200,
  2001.

\bibitem{bollobas}
B{\'e}la Bollob{\'a}s.
\newblock The isoperimetric number of random regular graphs.
\newblock {\em European J. Combin.}, 9(3):241--244, 1988.

\bibitem{BHR2010:speed}
M.~Broom, C.~Hadjicrysanthou, and J.~Rycht{\'a}{\v{r}}.
\newblock Evolutionary games on graphs and the speed of the evolutionary
  process.
\newblock {\em Proceedings of the Royal Society~A}, 466(2117):1327--1346, 2010.

\bibitem{BR2008:fixprob}
M.~Broom and J.~Rycht{\'a}{\v{r}}.
\newblock An analysis of the fixation probability of a mutant on special
  classes of non-directed graphs.
\newblock {\em Proceedings of the Royal Society~A}, 464(2098):2609--2627, 2008.

\bibitem{Buser}
Peter Buser.
\newblock Cubic graphs and the first eigenvalue of a {R}iemann surface.
\newblock {\em Math. Z.}, 162(1):87--99, 1978.

\bibitem{arrays}
M~Cemil~Azizo\u{g}lu and \"{O}mer E\u{g}ecio\u{g}lu.
\newblock The bisection width and the isoperimetric number of arrays.
\newblock {\em Discrete Appl. Math.}, 138(1--2):3--12, 2004.

\bibitem{DGMRSS2013:superstars}
Josep D{\'\i}az, Leslie~Ann Goldberg, George~B. Mertzios, David Richerby,
  Maria~J. Serna, and Paul~G. Spirakis.
\newblock On the fixation probability of superstars.
\newblock {\em Proceedings of the Royal Society~A}, 469(2156):20130193, 2013.

\bibitem{OurSODA}
Josep D\'{\i}az, Leslie~Ann Goldberg, George~B. Mertzios, David Richerby,
  Maria~J. Serna, and Paul~G. Spirakis.
\newblock Approximating fixation probabilities in the generalized {M}oran
  process.
\newblock {\em Algorithmica}, to appear.

\bibitem{Gin2000:GT}
Herbert Gintis.
\newblock {\em Game Theory Evolving: A Problem-Centered Introduction to
  Modeling Strategic Interaction}.
\newblock Princeton University Press, 2000.

\bibitem{GS2001:Probability}
G.~R. Grimmett and D.~R. Stirzaker.
\newblock {\em Probability and Random Processes}.
\newblock Oxford University Press, 3rd edition, 2001.

\bibitem{HV2011:fixation}
B.~Houchmandzadeh and M.~Vallade.
\newblock The fixation probability of a beneficial mutation in a geographically
  structured population.
\newblock {\em New Journal of Physics}, 13:073020, 2011.

\bibitem{fixrevised}
A.~{Jamieson-Lane} and C.~{Hauert}.
\newblock Fixation probabilities on superstars, revisited and revised.
\newblock ArXiv 1312.6333v2, 2014.

\bibitem{KKT2003:Influence}
David Kempe, Jon Kleinberg, and {\'E}va Tardos.
\newblock Maximizing the spread of influence through a social network.
\newblock In {\em Proc. 9th ACM International Conference on Knowledge Discovery
  and Data Mining}, pages 137--146. ACM, 2003.

\bibitem{LHN}
Erez Lieberman, Christoph Hauert, and Martin~A. Nowak.
\newblock Evolutionary dynamics on graphs.
\newblock {\em Nature}, 433(7023):312--316, 2005.

\bibitem{Lig1999:IntSys}
Thomas~M. Liggett.
\newblock {\em Stochastic Interacting Systems: Contact, Voter and Exclusion
  Processes}.
\newblock Springer, 1999.

\bibitem{MNRS2013:suppressors}
George~B. Mertzios, Sotiris~E. Nikoletseas, Christoforos Raptopoulos, and
  Paul~G. Spirakis.
\newblock Natural models for evolution on networks.
\newblock {\em Theoretical Computer Science}, 477:76--95, 2013.

\bibitem{MS2013:strongbounds}
George~B. Mertzios and Paul~G. Spirakis.
\newblock Strong bounds for evolution in networks.
\newblock In {\em Proc. 40th International Colloquium on Automata, Languages
  and Programming (ICALP 2013)}, volume 7966 of {\em LNCS}, pages 657--668.
  Springer, 2013.

\bibitem{MU}
Michael Mitzenmacher and Eli Upfal.
\newblock {\em Probability and Computing: Randomized Algorithms and
  Probabilistic Analysis}.
\newblock Cambridge University Press, 2005.

\bibitem{Mohar}
Bojan Mohar.
\newblock Isoperimetric numbers of graphs.
\newblock {\em Journal of Combinatorial Theory, Series B}, 47(3):274--291,
  1989.

\bibitem{Moran58}
P.~A.~P. Moran.
\newblock Random processes in genetics.
\newblock {\em Proceedings of the Cambridge Philosophical Society},
  54(1):60--71, 1958.

\bibitem{Norris}
James~R. Norris.
\newblock {\em Markov Chains}.
\newblock Cambridge University Press, 1998.

\bibitem{SR2011:fastfix}
Paulo Shakarian and Patrick Roos.
\newblock Fast and deterministic computation of fixation probability in
  evolutionary graphs.
\newblock In {\em Proc. 6th International Conference on Computational
  Intelligence and Bioinformatics}, pages 753--012. ACTA Press, 2011.

\bibitem{SRJ2012:review}
Paulo Shakarian, Patrick Roos, and Anthony Johnson.
\newblock A review of evolutionary graph theory with applications to game
  theory.
\newblock {\em Biosystems}, 107:66--80, 2012.

\bibitem{TFSN2004:game}
Christine Taylor, Drew Fudenberg, Akira Sasaki, and Martin~A. Nowak.
\newblock Evolutionary game dynamics in finite populations.
\newblock {\em Bulletin of Mathematical Biology}, 66(6):1621--1644, 2004.

\end{thebibliography}

\end{document}